\documentclass[11pt,times]{article}
\usepackage[utf8]{inputenc}
\usepackage{amssymb,amsmath}
\usepackage{algorithm}
\usepackage{algpseudocode}
\usepackage{amsthm}
\usepackage{enumerate}
\usepackage[font={scriptsize,it}]{caption}
\usepackage{subcaption}
\allowdisplaybreaks
\usepackage[T1]{fontenc}

\usepackage{multirow}
\usepackage{algpseudocode}

\usepackage{float}
\newfloat{algorithm}{t}{lop}

\usepackage{tikz}
\usetikzlibrary{decorations.pathmorphing}
\usetikzlibrary{shapes,fit}
\usetikzlibrary{backgrounds}

\usepackage{bm}

\usepackage[font=scriptsize,skip=5pt]{caption}
\usepackage[inline]{enumitem}

\newtheorem{theorem}{Theorem}[section]
\newtheorem{corollary}[theorem]{Corollary}

\newtheorem{lemma}[theorem]{Lemma}

\newtheorem*{theorem*}{Theorem}
\newtheorem*{corollary*}{Corollary}
\newtheorem*{lemma*}{Lemma}
\newtheorem*{observation*}{Observation}

\newtheorem{observation}[theorem]{Observation}

\newcommand{\Xomit}[1]{}
\newcommand{\vone}{\vspace{.1in}}

\newcommand{\currentSD}{{\it flag-$d^*$}}

\newcommand{\rco}{resulting in a contradiction}
\newcommand{\obs}{Observation}
\newcommand{\lv}{$list_v$ }
\newcommand{\dummy}{non-SP }
\newcommand{\la}{\leftarrow}

\newcommand{\congest}{{\sc Congest}}

\newcommand{\hide}[1]{}

\newcommand{\boi}{\begin{enumerate}}
\newcommand{\eoi}{\end{enumerate}}
\newcommand{\bii}{\begin{itemize}}
\newcommand{\eii}{\end{itemize}}

\renewcommand\labelenumi{(\roman{enumi})}
\renewcommand\theenumi\labelenumi

\algtext*{EndWhile}
\algtext*{EndFor}
\algtext*{EndIf}

\advance \topmargin by -\headheight
\advance \topmargin by -\headsep
\evensidemargin \oddsidemargin
\begin{document}

\title{A Deterministic Distributed Algorithm for Weighted All Pairs Shortest Paths Through Pipelining}
\author{Udit Agarwal $^{\star}$ and Vijaya Ramachandran\thanks{Dept. of Computer Science, University of Texas, Austin TX 78712. Email: {\tt udit@cs.utexas.edu, vlr@cs.utexas.edu}. This work was supported in part by NSF Grant CCF-1320675. The first author's research was also partially supported by a UT Austin Graduate School Summer Fellowship.}}

\maketitle
\begin{abstract}
We present a new pipelined approach to compute all pairs shortest paths (APSP) in a 
directed graph with nonnegative integer edge weights (including zero weights) in the \congest{} model 
in the distributed setting. Our 
deterministic distributed  algorithm computes  shortest paths of
distance at most $\Delta$ for all pairs of vertices in at most $2 n \sqrt{\Delta} + 2n$ rounds, and more generally, it computes $h$-hop shortest paths for $k$ sources in
$2 \sqrt{nkh} + n + k$ rounds. 
The algorithm is  simple, and it has  some novel features and a nontrivial analysis.
It uses only the directed edges in the graph for communication.
This algorithm can be used  as a base within asymptotically faster algorithms that
 match or improve on the current best deterministic bound of $\tilde{O}(n^{3/2})$ rounds for this problem
 when edge weights
are $O(n)$ or shortest path distances are $\tilde{O}(n^{3/2})$. These latter results are presented
in a companion paper~\cite{AR18b}.
\end{abstract}

\section{Introduction}\label{sec:intro}

\vspace{-.03in}

 Designing distributed algorithms for various network and graph problems 
such as shortest paths~\cite{ARKP18,HNS17,LP13,PR18,KN18}
is
a extensively studied area of research.
The {\sc Congest} model (described in Sec~\ref{sec:congest})
is a  
widely-used
model for these algorithms,
see~\cite{ARKP18,Elkin17,HNS17,LP13}.
In this paper we consider  distributed algorithms for the 
 computing all pairs shortest paths (APSP) and  related problems in a graph  with non-negative edge weights
 in the {\sc Congest} model.

In sequential computation, shortest paths can be computed much faster in graphs with non-negative
edge-weights (including zero weights) using the classic Dijkstra's algorithm~\cite{Dijkstra59} than in
graphs with
negative edge-weights. Additionally, negative edge-weights raise the possibility
of negative weight cycles in the graph, which usually do not
occur in practice, and hence are not modeled by real-world weighted graphs. Thus, in the distributed
setting, it is of importance to design fast shortest path algorithms that can handle non-negative
edge-weights, including edges of weight zero.

 The presence of zero weight edges creates challenges in the design of distributed algorithms  
as  observed in~\cite{HNS17}.
 (We review related work in Section~\ref{sec:related}.)
 One approach used for positive integer edge weights is to replace an edge of weight $d$ with $d$ 
 unweighted edges and then run an unweighted APSP algorithm such as~\cite{LP13,PR18} on this modified graph.
 This approach is used in approximate APSP algorithms~\cite{Nanongkai14, LP15}.
 However such an approach fails when zero weight edges may be present.   
 There are a few known algorithms that can handle zero weights, such as the $\tilde{O}(n^{5/4})$-round randomized
 APSP algorithm of Huang et al.~\cite{HNS17} (for polynomially bounded non-negative
 integer edge weights) and the
 $\tilde{O}(n^{3/2})$-round deterministic APSP algorithm of Agarwal et al.~\cite{ARKP18} (for graphs with arbitrary edge weights).
 
 \vspace{-.05in}
 
 \subsection{Our Results}
 
 \vspace{-.03in}
 
 We present a new pipelined approach for computing APSP and related problems on  an $n$-node graph $G=(V,E)$ with non-negative edge 
 weights $w(e), e\in E$, (including zero weights). 
 Our results hold for both directed and undirected graphs and we will assume
 w.l.o.g. that $G$ is directed.  Our distributed algorithm uses only the directed edges for
 communication, while the algorithms in~\cite{HNS17,ARKP18} need to use the underlying
 undirected graph as the communication network.

 \vone
 \noindent
 {\bf Our Pipelined APSP Algorithm for Weighted Graphs.}
An $h$-hop shortest path from $u$ to $v$ in $G$
is a path from $u$ to $v$ of minimum weight among all paths with at most $h$ edges (or {\it hops}).
The central algorithm we present is for computing 
  $h$-hop APSP,  or more generally, the 
$h$-hop $k$-sources 
shortest path problem
($(h,k)$-SSP),
 with an additional constraint that the shortest paths have distance at most $\bigtriangleup$ in $G$.
  We also compute an  $(h,k)$-SSP tree for each source, which contains  an $h$-hop shortest path
   from the source to every other vertex to which
  there exists a path with weight at most $\bigtriangleup$.
In the case of multiple $h$-hop shortest paths from a source $s$ to a vertex $v$, this tree contains the path with the smallest number of hops, breaking any further ties by choosing the predecessor vertex with smallest ID.
Our algorithm (Algorithm~\ref{alg2} in Section~\ref{sec:h-hopk-SSP}) is compact and easy to
implement, and has no large hidden constant factors in its bound on the number of rounds.
It can be viewed as a (substantial) generalization of the pipelined method for unweighted APSP
given in~\cite{PR18}, which is a refinement of~\cite{LP13}. 
Our algorithm uses key values that depend on both the weighted distance and the
hop length of a path, and it can store multiple distance values for a source at a given
node, with the guarantee that the shortest path distance will be identified. 
 This algorithm (Algorithm~\ref{alg2}) 
 achieves the bounds in the following theorem.

 \begin{theorem}\label{thm:alg2}
 Let $G=(V,E)$ be a directed or undirected edge-weighted graph, where all edge weights are 
 non-negative integers (with zero-weight edges allowed).
 The following 
 deterministic bounds can be obtained  in the \congest{} model for shortest path distances at most 
 $\bigtriangleup$.\\
 (i) \emph{$(h,k)$-SSP} in $2 \sqrt{\bigtriangleup kh} + k + h $ rounds.\\
 (ii) \emph{APSP} in  $2n \sqrt{\bigtriangleup} + 2n$ rounds.\\
 (iii) \emph{$k$-SSP} in $2 \sqrt{\bigtriangleup kn} + n + k$ rounds.
 \end{theorem}

\vone
\noindent
{\bf Follow-up Results.}
 In a companion paper~\cite{AR18b}, we build on this pipelined algorithm to present several improved results, which we summarize here.
 We improve on the 
 bounds given in $(ii)$ and $(iii)$ of Theorem~\ref{thm:alg2}
by combining our pipelined Algorithm~\ref{alg2} with a modified version of the 
 APSP algorithm in~\cite{ARKP18} to obtain  the bounds
 stated in the following Theorems~\ref{thm:algkSSPEdgeBound} and \ref{thm:algkSSP}.

\begin{theorem} \label{thm:algkSSPEdgeBound}
\cite{AR18b} Let $G=(V,E)$ be a directed or undirected edge-weighted graph, 
where all edge weights are 
 non-negative integers bounded by $\lambda$ (with zero-weight edges allowed).
 The following deterministic bounds can be obtained in the \congest{} model.\\
 (i) APSP  in $O(\lambda^{1/4}\cdot  n^{5/4} \log^{1/2} n)$ rounds.\\
 (ii) $k$-SSP in $O(\lambda^{1/4}\cdot  nk^{1/4} \log^{1/2} n)$ rounds.
 \end{theorem}

\begin{theorem} \label{thm:algkSSP}
\cite{AR18b}
Let $G=(V,E)$ be a directed or undirected edge-weighted graph, where all edge weights are 
non-negative integers (with zero edge-weights allowed),
 and the shortest path distances are bounded by $\bigtriangleup$.
 The following 
 deterministic bounds can be obtained in the \congest{} model.\\
 (i) APSP  in $O(n (\bigtriangleup \log^2 n)^{1/3})$ rounds.\\
 (ii) $k$-SSP in $O((\bigtriangleup kn^2 \log^2 n)^{1/3})$ rounds.
 \end{theorem}

The results in Theorem~\ref{thm:algkSSPEdgeBound} and \ref{thm:algkSSP}  improve on the 
$\tilde{O}(n^{3/2})$ deterministic APSP bound of Agarwal et al.~\cite{ARKP18} for significant ranges of values 
for both $\lambda$ and $\Delta$, as stated below.

   \begin{corollary} \label{cor:lambda}
   \cite{AR18b}
Let $G=(V,E)$ be a directed or undirected edge-weighted graph with non-negative edge weights 
(and zero-weight edges allowed). 
 The following deterministic bounds hold for 
 the \congest{} model for $1\geq \epsilon\geq 0$. \\
 (i) If the edge weights are bounded by $\lambda = n^{1-\epsilon}$, then APSP can be computed in $O(n^{3/2 - \epsilon/4}\log^{1/2} n)$ rounds.\\
 (ii) For shortest path distances bounded by $\Delta = n^{3/2-\epsilon}$, APSP can be computed in $O(n^{3/2 - \epsilon/3} \log^{2/3} n)$ rounds.
 \end{corollary}

 The corresponding bounds for the weighted $k$-SSP problem are: $O(n^{5/4-\epsilon/4}k^{1/4}\log^{1/2} n)$ (when $\lambda = n^{1-\epsilon}$) and $O(n^{7/6 - \epsilon/3}k^{1/3}\log^{2/3} n)$ (when $\Delta = n^{3/2-\epsilon}$).
 Note that the result in $(i)$ is independent of the value of $\Delta$ (depends only on $\lambda$)
 and the result in $(ii)$ is independent of the value of $\lambda$ (depends only on $\Delta$).

  Our pipelined technique can be adapted to give simpler methods for some procedures in the randomized 
 distributed  weighted APSP algorithms in Huang et al.~\cite{HNS17}. In~\cite{AR18b}  we present simple deterministic algorithms that
 match the congest and dilation bounds in~\cite{HNS17} for two of the three procedures used there: the {\it short-range} and {\it short-range-extension} algorithms. Those simplified algorithms are both obtained using a streamlined single-source version of our pipelined APSP algorithm (Algorithm~\ref{alg2}).
  Several other results for distributed computation of shortest paths are presented in~\cite{AR18b}.
 
 After discussing the {\sc Congest} model and related reults, the rest of this paper will present
 our new pipelined algorithm.

\subsection{Congest Model}	\label{sec:congest}

In the {\sc Congest} model,
there are $n$ independent processors 
interconnected in a network by bounded-bandwidth links.
We refer to these processors as nodes and
the links as edges.
This network is modeled by graph $G = (V,E)$
where $V$ refers to the set of processors and
$E$ refers to the set of links between the processors.
Here $|V| = n$ and $|E| = m$.

Each node is assigned a unique ID  
between 1 and $poly(n)$ and
has infinite computational power.
Each node has limited topological knowledge 
and only knows about its incident edges.
For the weighted APSP problem we consider, 
 each edge 
has a positive
or zero
  integer weight that can be represented with $B= O(\log n)$ bits.
Also if the edges are directed, 
the corresponding communication channels are bidirectional
and hence the communication network can be represented 
by the underlying undirected graph $U_G$ of $G$
(this is also considered in~\cite{HNS17,PR18,GL18}). It turns out that our basic pipelined algorithm does not need this
feature though our faster algorithm does.

The computation proceeds in rounds. In each round each processor can send a message of 
size $O(\log n)$ along edges incident to it, and it receives the messages sent to it in the previous
round. The model allows a node to send different message along different edges though we do not
need this feature in our algorithm.
The performance of an algorithm in the \congest{} model is measured by its
round complexity, which is the worst-case
number of rounds of distributed communication.
\vspace{-.03in}

\subsection{Related Work}\label{sec:related}

\noindent
{\bf Weighted APSP.}
The current best bound for the weighted APSP problem is due to the randomized algorithm of Huang et al.~\cite{HNS17}
that runs in $\tilde{O}(n^{5/4})$ rounds.
This algorithm works for graphs with polynomially bounded integer edge weights (including
zero-weight edges), and the result holds with w.h.p. in $n$.
For graphs with arbitrary edge weights, the recent result of Agarwal et al.~\cite{ARKP18} gives a deterministic APSP
algorithm that runs in $\tilde{O}(n^{3/2})$ rounds.
This is the current best bound (both deterministic and randomized) for graphs with arbitrary edge weights as well as
 the best deterministic bound for graphs with integer edge weights.
 
 In a companion paper~\cite{AR18b} we build on the pipelined algorithm we present here to 
obtain an algorithm for non-negative integer edge-weights (including zero-weighted
  edges) that runs in $\tilde{O}(n \bigtriangleup^{1/3})$ rounds 
where the shortest path distances are at most $\bigtriangleup$
and in $\tilde{O}(n^{5/4}\lambda^{1/4})$ rounds when the edge weights are bounded by $\lambda$.
This result improves on the $\tilde{O}(n^{3/2})$ deterministic APSP bound of Agarwal et al.~\cite{ARKP18} 
when either edge weights are at most $n^{1-\epsilon}$ or
 shortest path distances are at most $n^{3/2 - \epsilon}$, for any $\epsilon > 0$.
 In~\cite{AR18b} we also give an improved randomized algorithm for APSP in graphs with arbitrary edge weights that
 runs in $\tilde{O}(n^{4/3})$ rounds, w.h.p. in $n$.

\vspace{.03in}

\noindent
{\bf  Weighted $k$-SSP.}
 The current best bound for the weighted $k$-SSP problem is due to the Huang et al's~\cite{HNS17} randomized
 algorithm that runs in $\tilde{O}(n^{3/4}\cdot k^{1/2} + n)$ rounds.
 This algorithm is also randomized and only works for graphs with integer edge weights.
 The deterministic APSP algorithm in~\cite{ARKP18} can be shown to give an
$O(n \cdot \sqrt{k \log n})$ round deterministic algorithm for $k$-SSP.
In this paper, we present a deterministic algorithm for non-negative (including zero) integer edge-weighted graphs
that runs in $\tilde{O}((\bigtriangleup \cdot n^2 \cdot k)^{1/3})$ rounds where the shortest path distances are at
 most $\bigtriangleup$ and in $\tilde{O}((\lambda k)^{1/4}n)$ rounds when the edge weights are bounded by
 $\lambda$.

\vspace{-.03in}
\section{Overview of the $O(n\sqrt{\bigtriangleup})$ Pipelined Algorithm for APSP}\label{sec:overview}

\vspace{-.025in}

 The starting point for our weighted APSP algorithm is the distributed algorithm for
 {\it unweighted} APSP in~\cite{PR18}, which is a streamlined variant of an earlier
 APSP algorithm~\cite{LP13}.
 This unweighted APSP algorithm is very simple: each source initiates its distributed BFS in round 1.  Each node $v$
 retains the best (i.e., shortest) distance estimate it has received for each source, and  stores 
 these estimates in sorted order (breaking ties by source id). Let $d(s)$ (or $d_v(s)$) denote the
 shortest distance estimate for source $s$ at $v$ and let $pos(s)$ be its position in
 sorted order
($pos(s) \geq 1$).
   In a general round $r$, node $v$
 sends out a shortest distance estimate $d(s)$ if $r= d(s)  + pos(s)$. Since $d(s)$
 is nondecreasing and $pos(s)$ is increasing, there will be at most one $d(s)$ at $v$ that
 can satisfy this condition. It is shown 
 in~\cite{PR18} that  if the current best  distance estimate $d(s)$ for a source $s$ reaches $v$  in  round $r$ then
 $r <d(s) + pos(s)$.  Since $d(s) <n$ for any source
 $s$ and $pos(s)$ is at most $n$, shortest path values for all sources arrive at any given node
 $v$ in less than $2n$ rounds. 
 
 In the class of graphs that we deal with, $d(s)$ is at most $\bigtriangleup$ for all $s,v$,  
 and it appears plausible that the above pipelining method would apply here as well.
 Unfortunately, this does not hold 
 since we allow zero weight edges in the  graph.
  The key to
 the guarantee that a $d(s)$ value arrives at $v$ before round $d(s) + pos(s)$ in the 
 unweighted case in~\cite{PR18} is that
 the predecessor $y$ that sent its $d_y(s)$ value to $v$ must have had
 $d_y(s) = d_v(s) -1$. (Recall that in the unweighted case, $d_y(s)$ is simply the hop-length of the path taken from $s$ to $y$.)
 If we have zero-weight edges this guarantee no longer holds for the weighted path length, and it
 appears that the key property of the unweighted pipelining methodogy no longer applies. Since edge weights larger than 1 are also possible
 (as long as no shortest path distance exceeds $\bigtriangleup$) we also have the property that the hop length of a path can be either greater than or less
 than its weighted distance.

\subsection{Our $(h,k)$-SSP algorithm}

 Algorithm~\ref{alg2} in the next section
 is our pipelined algorithm for a
directed 
   graph $G=(V,E)$ with non-negative
    edge-weights. The input is
 $G$, together with the subset $S$ of $k$ vertices  for which we need to compute $h$-hop shortest path trees. An innovative feature
 of this algorithm is that the key $\kappa$ it uses for a path is not its weighted distance, but a function of {\it both} its hop length $l$
and its weighted distance $d$.  More specifically, $\kappa = d \cdot \gamma + l$, where 
$\gamma = \sqrt{kh/\Delta}$. This allows 
the key to inherit some of the properties
from the algorithms in~\cite{LP13,PR18} through the fact that the hop length is part of
$\kappa$'s value, while also retaining the weighted distance which is the actual 
value that needs to be computed.

The new key $\kappa$ by itself is not sufficient to adapt the algorithm for unweighted
APSP in~\cite{PR18} to the weighted case. In fact, the use of $\kappa$ can
complicate the computation since one can have two paths from $s$ to $v$, with
weighted distances $d_1 < d_2$, and yet for the associated keys one could have
$\kappa_1 > \kappa_2$ (because the path with the smaller weight can have a larger
hop-length).  Our algorithm handles this with another unusual feature:
 it may maintain several (though not all) of the key values it receives, and may also
send out several key values, even some that it knows cannot correspond to a shortest distance. These features are incorporated into
a carefully tailored algorithm  that terminates in $O(\sqrt{\bigtriangleup kh})$ rounds with all $h$-hop shortest path distances from the $k$ sources computed.

 It is not difficult to show that eventually every  shortest path distance key arrives at $v$ for each source from which $v$ is
 reachable when Algorithm~\ref{alg2} is executed. 
 In order to establish the bound on the number of rounds, we show  that our pipelined algorithm maintains two important invariants:
 
 \boi
 \item []{\bf Invariant 1:} If an entry $Z$ is added to $list_v$ in round $r$, then  
 $r < \lceil Z.\kappa + pos(Z) \rceil$, where $Z.\kappa$ is $Z's$ key value.
 
  \item[] {\bf Invariant 2:} The number of entries for a given source $s$ at $list_v$ is at most $\sqrt{\bigtriangleup h/k} + 1$.
  
  \eoi
  
Invariant 1 is the natural generalization of the unweighted algorithms~\cite{LP13,PR18}
 for the key $\kappa$ that we use. On the other hand, to the best of our knowledge, Invariant 2 
 has not been used before, nor has the notion of storing multiple 
 paths or entries
 for the same source at a
 given node.
By Invariant 2, the number of entries in any list is at most $\sqrt{\bigtriangleup kh} +k$, so 
$pos(Z) \leq \sqrt{\bigtriangleup kh} +k$ for every list at every round. Since the value of any $\kappa$ is at most 
$\bigtriangleup \cdot\gamma + h$, by Invariant 1 every  entry is received by round 
$2 \sqrt{\bigtriangleup kh} + k + h$. 
We give the details in the next section, starting with  a step-by-step description
 of Algorithm~\ref{alg2} followed by its analysis.

  \section{The Pipelined $(h,k)$-SSP Algorithm}	\label{sec:h-hopk-SSP}

 We now describe Algorithm~\ref{alg2}.
 Recall that the key value we use for a path $\pi$ is $\kappa = d \cdot \gamma + l$, 
 where $\gamma = \sqrt{kh/\Delta}$, $d$ is the weighted path length, 
 and $l$ is the hop-length of $\pi$.
 At each node $v$ our algorithm maintains a list, $list_v$, of the  
 entries and associated data
 it has retained. 
Each element $Z$ on $list_v$ is of the form $Z= (\kappa, d, l, x)$, where $x$ is the source vertex for  the path corresponding to
$\kappa$, $d$, and $l$. 
The elements on $list_v$ are ordered by
key value $\kappa$, with ties  first resolved by the value of $d$, and then by the label of the source vertex.  
 We use $Z.\nu$ to denote
the number of keys for source $x$ stored on $list_v$ at or
below $Z$. 
The position of an element $Z$ in $list_v$ is given by $pos(Z)$, which gives the number of elements at
or below $Z$ on $list_v$. 
If the vertex $v$ and the round $r$ are relevant to the discussion we will use
the notation $pos_v^r(Z)$, but we will remove either the subscript or the superscript (or both) if they
are clear from the context. 
We also have  a flag
$Z.$\currentSD which is set if $Z$ has the smallest $(d, \kappa)$ value among all entries for source $x$ (so $d$ is the shortest weighted distance from
$s$ to $v$ among all keys for $x$ on $list_v$). A summary of our notation is
in Table~\ref{table1}.

\begin{table*}
\scriptsize
\centering
\noindent
\caption{Notations} \label{table1} \vspace{-.02in}
\begin{tabular}{| c | l |}
\hline
\multicolumn{2}{|c|}{\sc Global Parameters:} \\
\hline
$S$ & set of sources\\
\hline
$k$ & number of sources, or $|S|$ \\
\hline
$h$ & maximum number of hops in a shortest path\\
\hline
$\bigtriangleup$ & maximum weighted distance of a shortest path \\
\hline
$n$ & number of nodes \\
\hline
$\gamma$ & parameter equal to $\sqrt{hk/\Delta}$\\
\hline
\hline
\multicolumn{2}{|c|}{Local Variables at node $v$:}\\
\hline
$d_x^*$ & current shortest path distance from $x$ to $v$; same as $d^*_{x,v}$ \\
\hline
$list_v$ & list at $v$ for storing the SP and non-SP entries \\
\hline
\end{tabular}
\quad
\begin{tabular}{| c | l |}
\hline
\multicolumn{2}{|c|}{Variables/Parameters for entry $Z = (\kappa, d, l, x)$ in $list_v$:}	\\
\hline
$\kappa$ & key for $Z$; $\kappa = d\cdot \gamma + h$ \\
\hline
$d$ & weight (distance) of the path associated with this entry	\\
\hline
$l$ & hop-length of the path associated with this entry		\\
\hline
$x$ & start node (i.e. source) of the path associated with this entry	\\	
\hline
$p$ & parent node of $v$ on the path associated with this entry	\\	
\hline
$\nu$ & number of entries for source $x$ at or below $Z$ in $list_v$ (not stored explicitly)	\\
\hline
\currentSD & flag to indicate if $Z$ is the current SP entry for source $x$	\\
\hline
$pos$ & position of $Z$ in $list_v$ in a round $r$; same as $pos^r$, $pos^r_v$	\\
\hline
$SP$ & shortest path \\
\hline
\end{tabular}
\end{table*}

 Initially, when round $r=0$, $list_v$ is empty unless $v$ is in the source
 set $S$. Each source vertex $x\in S$ places an element $(0,0,0,x)$ on its $list_x$ to indicate a path of
 weight 0 and  hop length 0 from $x$ to $x$, and
 $Z.$\currentSD is set to $true$.
 In Step~\ref{alg2:initd} of the Initialization round $0$, 
  node $v$
   initializes the distance from 
every source to $\infty$.
In Step~\ref{alg2:initAdd}  every source
 vertex initializes the distance from itself to $0$ and 
 adds the corresponding entry in its list.
 There are no Sends in round $0$.

  In a general round $r$, in Step~\ref{alg2:sendStart} of Algorithm~\ref{alg2},
 $v$ checks if $list_v$ contains an entry $Z$ with $\lceil Z.\kappa + pos_v(Z) \rceil = r$. If there is such an entry $Z$
 then $v$ sends $Z$ to its neighbors, along with $Z.\nu$ and $Z.flag$-$d^*$ in Step~\ref{alg2:sendEnd}.
Steps~\ref{alg2:receiveStart}-\ref{alg2:receiveEnd}
 describe the steps taken at $v$ after receiving a 
 set of incoming messages $I$ from its neighbors.
 In Step~\ref{alg2:Z} an entry $Z$ is created from an incoming message $M$, updated to reflect the $d$ and $l$ values at $v$.
 Step~\ref{alg2:checkZ} checks if $Z$ has
 a shorter distance than the current shortest path
 entry, $Z^*$, at $v$, or a shorter hop-length (if the distance is the same),
 or a  
parent with smaller ID (if both distance  and hop-length are same).
 And if so, then $Z$ is marked as SP in Step~\ref{alg2:setDummy}
 and is then inserted in $list_v$ in Step~\ref{alg2:insert1}.
 Otherwise, if $Z$ is a \dummy it is  inserted into $list_v$ in Step~\ref{alg2:checkIfEnough} only
 if the number of entries on $list_v$  for source $x$ with key $< Z.\kappa$ in $list_v$ is less than $Z^-.\nu$. This is the rule
 that decides if a received entry that is not the SP entry is inserted into $list_v$. 

 Steps~\ref{alg2:addZ}-\ref{alg2:remove} of procedure {\sc Insert}
 perform the addition of  a new entry $Z$ to $list_v$.
 In Step~\ref{alg2:addZ} $Z$ is inserted in $list_v$
 in the sorted order of $(\kappa,d,x)$. The algorithm then moves on to remove an existing entry for source $x$ 
 on $list_v$ if the condition  in Step~\ref{alg2:checkIfExists} holds.
 This condition 
checks if there is a \dummy entry above $Z$ in $list_v$.
If so then  the closest \dummy entry above $Z$ is removed
in Steps~\ref{alg2:findZ'}-\ref{alg2:remove}. 

\vspace{-.05in}

\begin{algorithm}[H]
\scriptsize
\caption*{\scriptsize {\sc Initialization:} Initialization procedure for Algorithm~\ref{alg2} at node $v$}
Input: set of sources $S$
\begin{algorithmic}[1]
\State {\bf for each} $x \in S$ {\bf do}   $d_x^* \leftarrow \infty$	\label{alg2:initd}
\State {\bf if} $v \in S$ {\bf then}  $d_v^* \leftarrow 0$; add an entry $Z = (0,0,0,v)$ to $list_v$; $Z.$\currentSD$\leftarrow true$ \vspace{-.05in} \label{alg2:initAdd}
\end{algorithmic}  \label{alg2:init}
\end{algorithm}

\vspace{-.2in}
\begin{algorithm}[H]
\scriptsize
\caption{\scriptsize Pipelined $(h,k)$-SSP algorithm at node $v$ for round $r$}
Input: A set of sources $S$
\begin{algorithmic}[1]
\State {\bf send [Steps~\ref{alg2:sendStart}-\ref{alg2:sendEnd}]:} {\bf if} there is an entry $Z$ with $\lceil Z.\kappa + pos^r_v (Z) \rceil = r$	\label{alg2:sendStart}
\State  \hspace{.1in} {\bf then} compute $Z.\nu$ and form the message $M=\langle Z,$ $Z.$\currentSD$, Z.\nu \rangle$ and send $M$ to all neighbors \vspace{0.05in}	\label{alg2:sendEnd}
\State {\bf receive [Steps~\ref{alg2:receiveStart}-\ref{alg2:receiveEnd}]:} let $I$ be the set of incoming messages	\label{alg2:receiveStart}
\For{{\bf each} $M \in I$}
	\State let $M=(Z^{-}=(\kappa^{-}, d^{-}, l^{-}, x), Z^-.$\currentSD$, Z^-.\nu)$ and let the sender be $y$.
	\State $\kappa \leftarrow \kappa^{-} + w (y,v)\cdot \gamma + 1$; $d \la d^{-} + w (y,v)$; $l \la l^{-} + 1$	\label{alg2:updateKey}
	\State $Z \leftarrow (\kappa, d, l, x)$; $Z.$\currentSD$\leftarrow false$; $Z.p \la y$     \label{alg2:Z}     \hspace{.1in} ($Z$ may be added to $list_v$ in Step ~\ref{alg2:insert1} or \ref{alg2:checkIfEnough})
	\State let $Z^*$ be the entry for $x$ in $list_v$ such that $Z^*.$\currentSD$ = true$, if such an entry exists (otherwise $d_x^* = \infty$)	\label{alg2:pickZ*}
	\If{$Z^-.$\currentSD$ = true$ and $l \leq h$ and $( \left( d  < d_x^* \right)$ or $( d = d_x^*$ and $Z.\kappa < Z^*.\kappa) \text{ or } \left(d = d^*_x \text{ and } Z.\kappa = Z^*.\kappa \text{ and } Z.p < Z^*.p \right)) $}	\label{alg2:checkZ}
		\State $d_x^* \leftarrow d$; $Z.$\currentSD$ \leftarrow true$;  $Z^*.$\currentSD$ \leftarrow false$ (if $Z^*$ exists)	\label{alg2:setDummy}
		\State  {\sc Insert}($Z$)   \label{alg2:insert1}
	\Else
		\State {\bf if} there are less than $Z^-.\nu$ entries for $x$ with $key \leq Z.\kappa$ {\bf then} {\sc Insert}($Z$)	 \vspace{-.05in} \label{alg2:checkIfEnough}
	\EndIf
\EndFor	\label{alg2:receiveEnd}
\end{algorithmic} \label{alg2}
\end{algorithm}

\vspace{-.2in}
\begin{algorithm}[H]
\scriptsize
\caption*{\scriptsize {\sc Insert}($Z$): Procedure for adding $Z$ to $list_v$}
\begin{algorithmic}[1]
\State insert $Z$ in $list_v$ in sorted order of $(\kappa, d,x)$ \label{alg2:addZ}
\If{$\exists$  an entry $Z^{'}$ for $x$ in $list_v$ such that $Z^{'}.$\currentSD$ = false$ and $pos(Z^{'}) > pos(Z)$}	\label{alg2:checkIfExists}
	\State find $Z^{'}$ with smallest $pos(Z^{'})$ such that $pos(Z^{'}) > pos(Z)$ and  $Z^{'}.$\currentSD$ = false$		\label{alg2:findZ'}
	\State remove $Z^{'}$ from $list_v$	\label{alg2:remove}
\EndIf
\end{algorithmic}  \label{alg2:insert}
\end{algorithm}

\vspace{-.1in}

Algorithm~\ref{alg2} performs these steps in successive rounds. We next analyze it for correctness and we
also show that it terminates with all shortest distances computed before round 
$r =\lceil 2\sqrt{\bigtriangleup kh} +k +h \rceil$.
  
  \vspace{-.05in}
  
\subsection{Correctness of Algorithm~\ref{alg2}}

We now provide a sketch for correctness of Algorithm~\ref{alg2}.
The complete proofs are in 
Appendix~\ref{sec:proofsAlg2}. 
The initial Observations and Lemmas given below establish
useful properties of an entry $Z$ in a $list_v$ and of $pos_v^r(Z)$ and
its relation to $pos_y^r (Z^-)$. We then present the key lemmas.
In Lemma~\ref{lemma:map}, we show that the collection of entries for a given source $x$ in $list_v$
can be mapped into $(d,l)$ pairs with non-negative $l$ values
such that $d=d^*$ for the shortest path entry, and the $d$ values
for all other entries are distinct and larger than $d^*$. (It turns out that we cannot simply use the
$d$ values already present in $Z$'s entries for this mapping since we could have
 two different entries for source $x$
on $list_v$, $Z_1$ and $Z_2$, that have the same $d$ value. )
Once we have Lemma~\ref{lemma:map} we
are able to bound the number of entries for a given source at $list_v$ by 
$\frac{h}{\gamma} + 1$
in Lemma~\ref{lemma:count}, and this establishes Invariant 2 (which is stated in Section~\ref{sec:overview}).
Lemma~\ref{lemma:sendBound} establishes Invariant 1.
In Lemma~\ref{lemma:sp} we establish that
all shortest path values reach node $v$. With these results in
hand, the final Lemma~\ref{lemma:hhopkSSP} for the round bound for 
computing $(h,k)$-SSP with shortest path distances at most $\bigtriangleup$
is readily
established, which then gives Theorem~\ref{thm:alg2}.

\begin{observation}	\label{obs:posremove}
Let $Z$ be an entry for a source $x \in S$ added to $list_v$ in round $r$. 
Then if $Z$ is removed 
from $list_v$
in a round $r' \geq r$, 
it was replaced by another entry for $x$, $Z'$,
such that $pos_v^{r'} (Z) > pos_v^{r'} (Z')$ and $Z.\kappa \geq Z'.\kappa$.
\end{observation}

\begin{lemma}	\label{lemma:pos}
Let $Z$ be an entry in $list_v$. Then $pos_v^{r'} (Z) \geq pos_v^r (Z)$ for all rounds $r' > r$, 
for which $Z$ exists in $v$'s list.
\end{lemma}

\begin{observation}	\label{obs:remove}
Let $Z$ be an  entry for source $x$
that was added to $list_v$.
If there exists a \dummy entry for $x$ above $Z$ in $list_v$,
then the closest \dummy entry above $Z$ will be removed.
\end{observation}

\begin{observation}	\label{obs:Znotadded}
Let $Z^{-}$ be an entry for source $x$ sent from  $y$ to $v$ in round $r$,
and let $Z$ be the corresponding entry created for possible addition to $list_v$ in 
Step~\ref{alg2:Z}  of
Algorithm~\ref{alg2}.  
If $Z$ is not added to $list_v$, then there is an entry $Z'\neq Z$ for source $x$
in $list_v$ with $Z'.$\currentSD$ = true$,
and there are at least $Z^-.\nu$ entries for $x$ with $key \leq Z.\kappa$ at the end of round $r$.
\end{observation}

\begin{lemma}\label{lemma:counts}
Let $Z$ be an entry for source $x$ that is present on $list_v$ in round $r$. Let $r'>r$, and let $c$ and $c'$ be
the number of entries for source $x$ on $list_v$ that have key value less than $Z$'s key value
in rounds $r$ and $r'$ respectively.
Then $c' \geq c$.
\end{lemma}

\vspace{-.1in}

Lemma~\ref{lemma:counts} holds for every round greater than $r$, even if
$Z$ is removed from $list_v$. 

\begin{lemma}	\label{lemma:countDummySet}
Let $Z^{-}$ be an entry for source $x$ sent from  $y$ to $v$ and
suppose the corresponding entry $Z$ (Step~\ref{alg2:Z} of Algorithm 2) is
added
to $list_v$ in round $r$.
Then 
there are at least $Z^-.\nu$ entries at or below $Z$ in $list_v$ for source $x$.
\end{lemma}
\vspace{-.2in}
\begin{proof}
Assume inductively that this result holds for all entries on $list_v$ and $list_y$ with key value at most $Z.\kappa$ at
all previous rounds and at $y$ in round $r$ as well.  
(It trivially holds initially.) 

Let $Z_1^{-}$ be the $(Z^-.\nu - 1)$-th entry for source $x$ in $list_y$. Since $Z_1^{-}$ has a key value smaller than $Z^-$ it was sent to $v$ in an earlier round $r'$. If the
corresponding entry $Z_1$
created for possible addition to $list_v$ in Step~\ref{alg2:Z}  of
Algorithm~\ref{alg2},
was inserted in $list_v$ then by inductive assumption 
there were at least $Z_1^-.\nu = Z^-.\nu - 1$ entries for $x$ at or below $Z_1$ in $list_v$.
And by Lemma~\ref{lemma:counts}
this holds for round $r$ as well and hence the result follows
since $Z$ is present above $Z_1$ in $list_v$.

And if $Z_1$ was not added to $list_v$ in round $r'$,
then by \obs ~\ref{obs:Znotadded}
 there were already $Z^-.\nu - 1$ 
 entries for $x$ 
with key $\leq Z_1.\kappa$
and by Lemma~\ref{lemma:counts} there are at least 
$Z^-.\nu - 1$ 
entries for $x$ with key $\leq Z_1.\kappa \leq Z.\kappa$ on $list_v$ at round $r$
and hence the result follows.
\end{proof}

\begin{lemma}	\label{lemma:ci}
Let $Z^-$ be an entry sent from  $y$ to $v$ in round $r$
and let $Z$ be the corresponding entry created for possible addition to $list_v$ in Step~\ref{alg2:Z}  of
Algorithm~\ref{alg2}.  
For each source $x_i \in S$,
let there be exactly $c_i$ entries for $x_i$ at or below $Z^-$ in $list_y$.
If $Z$ is added to $list_v$,
then for each $x_i \in S$,
there are at least $c_i$ entries for $x_i$ at or below $Z$ in $list_v$.
\end{lemma}

\begin{corollary}	\label{cor:pos2}
Let $Z^-$ be an entry sent from  $y$ to $v$ in round $r$
and let $Z$ be the corresponding entry created for possible addition to $list_v$ in Step~\ref{alg2:Z}  of
Algorithm~\ref{alg2}.  
If $Z$ is added to $list_v$, 
then $pos^r_y (Z^{-}) \leq pos^r_v (Z)$.
\end{corollary}
\begin{lemma}	\label{lemma:map}
Let $\mathcal{C}$ be the entries for a source $x \in S$ in $list_v$ in  round $r$. 
Then the entries in $\mathcal{C}$ can be mapped to $(d,l)$ pairs such that each
$l\geq 0$ and
each $Z \in \mathcal{C}$ is mapped to a distinct $d$ value with
$Z.\kappa = d\cdot \gamma + l$.
Also $d = d_x^*$ if $Z$ is a current shortest path entry, otherwise $d > d_x^*$.
\end{lemma}

\vspace{-.2in}
\begin{proof}
We will establish this result by induction on $j$, the number of entries in  $\mathcal{C}$.
For the base case, when $j=1$, we can map $d$ and $l$ to the pair in the single entry $Z$ since
$Z.\kappa = d \cdot \gamma +l$. Assume inductively that the result holds at $list_u$ for all nodes $u$ 
when the number of entries for $x$ is at most $j-1$. Consider the first time $|\mathcal{C}|$ becomes $j$
at $list_v$, and let this occur when node $y$ sends $Z^-$ to $v$ and this is updated and inserted as $Z$ 
in $list_v$ in round $r$.

If $Z$ is inserted as a new shortest path entry with distance value $d^*$, then the distinct $d$ values currently assigned to the $j-1$ entries
for source $x$ in $list_v$ must all be larger than $d^*$ hence we can simply assign the
 $d$ and $l$ values in $Z$ as its $(d,l)$
mapping.

If $Z$ is inserted as a \dummy entry then it is possible that the $d$ value in $Z$ has already been assigned to
one of the $j-1$ entries for source $x$ on $list_v$. If this is the case, consider the entries 
for source $x$  with key value at most $Z^-.\kappa$ in $list_y$ (at node $y$). 
By the check in Step~\ref{alg2:checkIfEnough} of Algorithm~\ref{alg2} 
 we know that there are $j$ such
values. Inductively these $j$ entries have $j$ distinct $d^-$ values assigned to them, and we transform these 
into $j$ distinct values for $list_v$ by adding $w (y,v)\cdot \gamma + 1$ to each of them. For at least one of these
$d^-$ values in $y$, call it $d_1^-$,  it must be the case that $d'=d_1^-+ w (y,v)\cdot \gamma + 1$ 
is not assigned to any of the $j-1$ entries for source $x$ below $Z$ in $list_v$. Let $Z^{-}_1$ be the entry in $y$'s list that is 
associated with distance $d_1^-$. 
It is readily seen
 that the associated $l$ value for $d'$ in $Z$ on $list_v$
must be greater than 0
 and the distance value $d' + w (y,v) \neq d_x^*$ 
 (see proof in Appendix \ref{sec:proofsAlg2}).
 So we can assign this $d$ value to $Z$.
 We do not need
  to consider \dummy entries above $Z$ since if there were one, the closest one above $Z$ would
  have been deleted and 
  $j-1$
   would not have increased to  $j$.

In the general case when the number of entries remains at $j$ after the insert we do need to consider
the possibility of the new value assigned to $Z$ being duplicated at an entry $Z'$ above $Z$. But here we
can assign to $Z'$ the $d$ value previously given to the removed entry (and the $l$ needed for 
$Z'.\kappa$ will be non-negative because the removed entry must have been below $Z'$ on $list_v$).
\end{proof}

\begin{lemma}	\label{lemma:belowSD}
Let $Z$ be the current shortest path distance entry for a source $x \in S$ in $v$'s list.
Then the number of entries for $x$ below $Z$ in $list_v$ is at most 
$h/\gamma$.
\end{lemma}

\vspace{-.2in}
\begin{proof}
By Lemma~\ref{lemma:map}, 
we know that the keys of all the entries for $x$ can be mapped to $(d,l)$ pairs
such that each entry is mapped to a distinct $d$ value and $l > 0$.

We have $Z.\kappa = d_x^*\cdot \gamma + l_x^*$, where $l_x^*$ is the hop-length of the shortest path from $x$ to $v$.
Let $Z^{''}$ be an entry for $x$ below $Z$ in $v$'s list. 
Then, $Z^{''}.\kappa \leq Z.\kappa$. 
It implies 
that $d^{''}\cdot \gamma \leq  d_x^*\cdot \gamma + (l_x^* -  l^{''}) <  d_x^*\cdot \gamma + h$
which gives 
$d^{''} <  d_x^* + h/\gamma $.
Since $d^{''} \geq d^*_x$,
there can be at most $h/\gamma$  entries for $x$ below $Z$
in $list_v$.
\end{proof}

\vspace{-.05in}

Using Lemma~\ref{lemma:belowSD} we can show that there are at most 
$h/\gamma + 1$
entries for source $x$ in $list_v$
in the case when the entry for  the shortest distance for $x$  is not the topmost entry in $list_v$. 

\begin{lemma}	\label{lemma:count}
For each source $x \in S$, 
$v$'s list has at most 
$h/\gamma + 1$
entries for $x$.
\end{lemma}

\vspace{-.05in}

In Lemmas~\ref{lemma:sendBound}-\ref{lemma:sp} we establish
an upper bound on the round $r$ by which $v$ receives a shortest path entry $Z^*$.
 
 \begin{lemma}	\label{lemma:sendBound}
If an entry $Z$ is added to $list_v$ in round $r$ then 
 $r  < \lceil Z.\kappa + pos^{r}_v (Z) \rceil$.
\end{lemma}
\vspace{-.2in}
\begin{proof}
The lemma holds in the first round since all entries have non-negative $\kappa$, any received entry has hop length at least 1, and the lowest
position is 1 so  for any entry $Z$  received by $v$ in round 1, $\lceil Z.\kappa + pos^1_v (Z) \rceil \geq  1+1 > 1$.

Let $r$ be the first round (if any) in which the lemma is violated, and let it occur 
when entry $Z$ is added to $list_v$.
So $r \geq  \lceil Z.\kappa + pos^r_v (Z) \rceil$. 
 Let $r_1 = \lceil Z.\kappa + pos^r_v (Z) \rceil$ (so $r_1\leq r$ by assumption).

Since $Z$ was added to $list_v$ in round $r$,  $Z^-$ was sent to $v$ by a node $y$ 
 in round $r$. So 
by Step~\ref{alg2:sendStart}
 $r= \lceil Z^-.\kappa + pos_y^r(Z^-) \rceil$. But $Z.\kappa > Z^-.\kappa$ and
 $pos_v^r(Z) \geq pos_y^r(Z^-)$, hence $r$ must be less than $\lceil Z.\kappa + pos^r_v (Z) \rceil$.
\end{proof}
   
   \begin{lemma}	\label{lemma:sp}
 Let  $\pi^*_{x,v}$ be a shortest path from source $x$ to $v$ with the minimum number of hops among
 $h$-hop
 shortest paths from $x$ to $v$. Let  $\pi^*_{x,v}$ have $l^*$ hops and
shortest path distance  $d^{*}_{x,v}$.
 Then $v$ 
receives
an entry 
  $Z^*= (\kappa, d^{*}_{x,v}, l^*, x)$ 
by round 
$r < \lceil Z^*.\kappa + pos_v^r (Z^*)\rceil$.
 \end{lemma}
  \vspace{-.2in}
 \begin{proof}
 If an entry 
  $Z^*= (\kappa, d^{*}_{x,v}, l^*, x)$
  is placed on $list_v$ by $v$
 then by Lemma~\ref{lemma:sendBound}
 it is received before round  $\lceil Z^*.\kappa + pos_v^r (Z^*) \rceil$ and hence
  it will be sent in round $r = \lceil Z^*.\kappa + pos_v^r (Z^*) \rceil$
 in Step~\ref{alg2:sendEnd}.
  It remains to show that an entry for path $\pi^*_{x,v}$   is received by $v$. We
 establish  this for all pairs $x,v$ by induction on key value $\kappa$. 
 
   If $\kappa = 0$, then it implies that the shortest path is the vertex $x$ itself and 
  thus the statement holds for $\kappa = 0$.
 Let us assume that the statement holds for all keys $< \kappa$ and consider the path $\pi^*_{x,v}$ with key 
 $\kappa = d^{*}_{x,v} \cdot \gamma + l^*$.
 
   Let $(y,v)$ be the last edge on the
  path $\pi^*_{x,v}$
  and let $\pi^*_{x,y}$ be the subpath of $\pi^*_{x,v}$ from $x$ to $y$.
  By construction the path $\pi^*_{x,y}$ is a shortest path from $x$ to $y$ and its hop length $l^* -1$ is the smallest
  among all shortest paths from $x$ to $y$. Hence by the inductive assumption an entry 
   $Z^{-}$ with  $Z^{-}.\kappa = d^*_{x,y} \cdot \gamma + l^* -1$ (which is strictly less than $Z^{*}.\kappa$)
   is received by $y$ before round $\lceil Z^{-}.\kappa + pos_y^{'} (Z^{-}) \rceil$ (by Lemma~\ref{lemma:sendBound})
   and is then sent to $v$ in round $r' = \lceil Z^{-}.\kappa + pos_y^{r'} (Z^{-}) \rceil$ in Step~\ref{alg2:sendEnd}.
  Thus $v$ adds the shortest path entry for $x$, $Z^*$,
  to $list_v$ by end of round $r'$.
   \end{proof} 
 
 \vspace{-0.05in}
 
In Lemma~\ref{lemma:hhopkSSP} we establish an upper bound on the round $r$ 
by which Algorithm~\ref{alg2} terminates.
  
   \vspace{-0.05in}

    \begin{lemma}	\label{lemma:hhopkSSP}
    Let $\Delta$ be the maximum shortest path distance in the $h$-hop paths.
   Algorithm~\ref{alg2} correctly computes the $h$-hop shortest path distances from each source $x \in S$ to each node $v \in V$
   by round 
 $\lceil \bigtriangleup\gamma + h + \bigtriangleup\cdot\gamma + k \rceil$.
   \end{lemma}
   \vspace{-.2in}
   \begin{proof}
   An $h$-hop shortest path has hop-length at most $h$ and 
   weight at most $\bigtriangleup$, 
   hence a key corresponding to a shortest path entry 
   will have value at most $\bigtriangleup\gamma + h$.
   Thus by Lemma~\ref{lemma:sp},
   for every source $x \in S$ every node $v \in V$
should have received the shortest path distance entry, $Z^*$,
   for source $x$ by round $r = \lceil \bigtriangleup\gamma + h + pos_v^r (Z^*) \rceil$.
   
   Now we need to bound the value of $pos_v^r (Z^*)$.
   By Lemma~\ref{lemma:count}, 
   we know that there are at most $h/\gamma + 1$ entries 
   for each source $x \in S$ in a node $v$'s list.
   Now as there are $k$ sources,
   $v$'s list has at most 
   $(h/\gamma+ 1)\cdot k \leq \gamma\cdot \bigtriangleup + k$ 
   entries,
   thus $pos_v^r (Z^*) \leq \gamma\cdot \bigtriangleup + k$ and
   hence $r \leq \lceil \Delta\gamma + h + \gamma\cdot \Delta + k \rceil$.
   \end{proof}

\vspace{-.05in}   
  Since $\gamma =  \sqrt{hk/\Delta}$, Lemma~\ref{lemma:hhopkSSP}  establishes 
  Theorem~\ref{thm:alg2}.
 
\section{Conclusion}

We have presented  a new approach to the distributed computation of shortest paths in a graph with
non-negative integer weights (including zero weights). Our deterministic pipelined distributed algorithms for
weighted shortest paths (both APSP, and for $k$ sources) is novel and very simple. Its asymptotic performance improves
on  $\tilde{O}(n^{3/2})$ rounds,  the current best deterministic distributed algorithm for this problem~\cite{ARKP18},
only in very special cases when shortest path distances are smaller than $n$. But the algorithm may be relevent even for larger
shortest path lengths since it is very simple, and has the very small constant factor  2 in the leading term. As noted in the introduction,
we have built on this algorithm to achieve several new results, including improved  deterministic APSP for moderately
large non-negative integer weights (including zero weights)~\cite{AR18b}.

 A major open problem left by our work is whether we can come up with a similar pipelining strategy when working with the
 scaled graph in Gabow's scaling technique~\cite{Gabow85}. 
 Our current pipelined algorithm assumes that all sources see the same 
 weight on each edge, while in the scaling algorithm each source sees a different edge weight on a given 
 edge.
 We could obtain a deterministic $\tilde{O}(n^{4/3})$-round
 APSP algorithm with non-negative polynomially bounded integer weights if our pipelined strategy can be
 made to work with
 Gabow's scaling technique~\cite{Gabow85}. While this can be handled with $n$ different SSSP computations in conjunction with the randomized
 scheduling result of Ghaffari~\cite{Ghaffari15}, it will be very interesting to see if a deterministic pipelined
 strategy could achieve the same result.

\bibliographystyle{abbrv}
\bibliography{references}

\appendix

\hide{
\section{Appendix}

\subsection{Correctness of Algorithm~\ref{alg2}}	\label{sec:proofsAlg2}
}

\section{Appendix: Correctness of Algorithm~\ref{alg2}}	\label{sec:proofsAlg2}

\underline{\bf Observations and Lemmas~\ref{obs:posremove}-\ref{lemma:countDummySet}:}
In the following Observations and Lemmas
we point out the key facts about an entry $Z$ in $list_v$
in our Algorithm~\ref{alg2}.
We use these in our proofs 
in this section.

\begin{observation*}{\bf \ref{obs:posremove}.}
Let $Z$ be an entry for a source $x \in S$ added to $list_v$ in round $r$. 
Then if $Z$ is removed 
from $list_v$
in a round $r' \geq r$, 
it was replaced by another entry for $x$, $Z'$,
such that $pos_v^{r'} (Z) > pos_v^{r'} (Z')$ and $Z.\kappa \geq Z'.\kappa$.
\end{observation*}
\vspace{-0.2in}
\begin{proof}
An entry is removed from $list_v$ only in Step~\ref{alg2:remove} of {\sc Insert}, and this occurs
at most once in round $r'$ (through a call from either Step~\ref{alg2:insert1} or Step~\ref{alg2:checkIfEnough} of Algorithm~\ref{alg2}). But
immediately before that removal an entry $Z'$ with a smaller value was inserted in $list_v$ in Step~\ref{alg2:addZ}
of {\sc Insert}.
\end{proof}

\begin{lemma*}{\bf \ref{lemma:pos}.}
Let $Z$ be an entry in $list_v$. Then $pos_v^{r'} (Z) \geq pos_v^r (Z)$ for all rounds $r' > r$, 
for which $Z$ exists in $v$'s list.
\end{lemma*}
\begin{proof}
If not, then it implies that there exists $Z'$ such that $Z'$ was below $Z$ in $v$'s list in round $r$ 
and was replaced by another entry $Z^{''}$ that was above $Z$ in a  round $r^{''}$ such that $r' \geq r^{''} > r$
and hence $pos_v^{r^{''}} (Z^{''}) > pos_v^{r^{''}} (Z^{'})$.
But by  Observation~\ref{obs:posremove} this cannot happen and 
thus resulting in a contradiction.
\end{proof}

\begin{observation*}{\bf \ref{obs:remove}.}
Let $Z$ be an  entry for source $x$
that was added to $list_v$.
If there exists a \dummy entry for $x$ above $Z$ in $list_v$,
then the closest \dummy entry above $Z$ will be removed from $list_v$.
\end{observation*}

\begin{proof}
This is immediate from Steps~\ref{alg2:addZ}-\ref{alg2:remove} of the procedure {\sc Insert}.
\end{proof}

\begin{observation*}{\bf \ref{obs:Znotadded}.}
Let $Z^{-}$ be an entry for a source $x$ sent from  $y$ to $v$ in round $r$,
and let $Z$ be the corresponding entry created for possible addition to $list_v$ in 
Step~\ref{alg2:Z}  of
Algorithm~\ref{alg2}.  
If $Z$ is not added to $list_v$, then there is an entry $Z'\neq Z$ in $list_v$ with $Z'.$\currentSD$ = true$,
and there are at least $Z^-.\nu$ entries for $x$ with $key \leq Z.\kappa$ at the end of round $r$.
\end{observation*}

\begin{proof}
This is immediate from 
Steps~\ref{alg2:pickZ*} and \ref{alg2:checkZ} of Algorithm~\ref{alg2} where the current entry $Z^*$ with $Z^*.$\currentSD$ = true$ is verified to have a 
shorter distance (or a smaller key if $Z$ and $Z^*$ have the same distance), and by
the check in Step~\ref{alg2:checkIfEnough}.
\end{proof}

\begin{observation}	\label{obs:SPlength}
Let $Z^*$ be a current SP entry for a source $x \in S$ present in $list_v$. Then $Z^*.l \leq h$.
\end{observation}
\begin{proof}
This is immediate from the check in Step~\ref{alg2:checkZ}.
\end{proof}

The above Observation should be contrasted with the fact
 that $list_v$ could contain entries $Z$ with $Z.l > h$, but only if  \currentSD$(Z)  = false$.
In fact it is possible that $list_v$ contains an entry $Z'\neq Z^*$ with $Z'.d=d^*$ and $l >h$ since such an
entry would fail the check in 
Step~\ref{alg2:checkZ} but could then be inserted in Step~\ref{alg2:checkIfEnough} of
Algorithm~\ref{alg2}.

\begin{lemma*}{\bf \ref{lemma:counts}.}
Let $Z$ be an entry for source $x$ that is present on $list_v$ in round $r$. Let $r'>r$, and let $c$ and $c'$ be
the number of entries for source $x$ on $list_v$ that have key value less than $Z$'s key value
in rounds $r$ and $r'$ respectively.
Then $c' \geq c$.
\end{lemma*}

\begin{proof}
If $c'<c$ then an entry for $x$ that was present below $Z$ in round $r$ must have been removed without having
another entry for $x$ being inserted below $Z$. But by Observation~\ref{obs:posremove} this is not possible since
any time an entry for source $x$ is removed from $list_v$ another entry for source $x$ with smaller key
value is inserted in $list_v$. 
\end{proof}

Lemma~\ref{lemma:counts} holds for every round greater than $r$, even if
$Z$ is removed from $list_v$. 
The following stronger lemma holds for rounds greater than
$r$ when $Z$ remains on $list_v$.

\begin{lemma}\label{lemma:countDummy}
Let $Z$ be a \dummy entry for source $x$ that is present on $list_v$ in round $r$. 
Let $r'>r$, and let $c$ and $c'$ be
the number of entries for source $x$ on $list_v$ that have key value less than $Z$'s key value
in rounds $r$ and $r'$ respectively.
Then $c' = c$.
\end{lemma}

\begin{proof}
If a new entry $Z^{'}$ with key $< Z.\kappa$ for $x$ is added, 
then by Observation~\ref{obs:remove} the closest non-SP entry for $x$
with key $> Z^{'}.\kappa$ must be removed from $list_v$ and
thus $c^{'} \leq c$.
Then using Lemma~\ref{lemma:counts} we have $c' = c$.
\end{proof}

\begin{lemma*}{\bf \ref{lemma:countDummySet}.}
Let $Z^{-}$ be an entry for source $x$ sent from  $y$ to $v$ and
suppose the corresponding entry $Z$ (Step~\ref{alg2:Z} of Algorithm 2) is
added
to $list_v$ in round $r$.
Then 
there are at least $Z^-.\nu$ entries at or below $Z$ in $list_v$ for source $x$.
\end{lemma*}
\vspace{-.25in}
\begin{proof}
Let us assume inductively that this result holds for all entries on $list_v$ and $list_y$ with key value at most $Z.\kappa$ at
all previous rounds and at $y$ in round $r$ as well. 
(It trivially holds initially.) 

Let $Z_1^{-}$ be the $(Z^-.\nu - 1)$-th entry for source $x$ in $list_y$. Since $Z_1^{-}$ has a key value smaller than $Z^-$ it was sent to $v$ in an earlier round $r'$. If the
corresponding entry $Z_1$
created for possible addition to $list_v$ in Step~\ref{alg2:Z}  of
Algorithm~\ref{alg2},
was inserted in $list_v$ then by inductive assumption 
there were at least $Z_1^-.\nu = Z^-.\nu - 1$ entries for $x$ at or below $Z_1$ in $list_v$.
And by Lemma~\ref{lemma:counts}
this holds for round $r$ as well and hence the result follows
since $Z$ is present above $Z_1$ in $list_v$.

And if $Z_1$ was not added to $list_v$ in round $r'$,
then by \obs ~\ref{obs:Znotadded}
 there were already $Z^-.\nu - 1$ 
 entries for $x$ 
with key $\leq Z_1.\kappa$
and by Lemma~\ref{lemma:counts} there are at least 
$Z^-.\nu - 1$ 
entries for $x$ with key $\leq Z_1.\kappa \leq Z.\kappa$ on $list_v$ at round $r$
and hence the result follows.
\end{proof}

\vone
\underline{\bf Establishing $\mathbf{pos_y^r (Z^-) \leq pos_v^r (Z)}$:}
For an entry $Z^-$ sent from $y$ to $v$ such that
$Z$ is the corresponding entry created for possible addition to $list_v$ in
 Step~\ref{alg2:Z}  of
Algorithm~\ref{alg2},
in Lemma~\ref{lemma:ci} and Corollary~\ref{cor:pos2}
we establish that
if $Z$ is added to $list_v$ then
$pos_y^r (Z^-) \leq pos_v^r (Z)$,
which is an important property of $pos$.

\begin{lemma*}{\bf \ref{lemma:ci}.}
Let $Z^-$ be an entry sent from  $y$ to $v$ in round $r$
and let $Z$ be the corresponding entry created for possible addition to $list_v$ in Step~\ref{alg2:Z}  of
Algorithm~\ref{alg2}.  
For each source $x_i \in S$,
let there be exactly $c_i$ entries for $x_i$ at or below $Z^-$ in $list_y$.
If $Z$ is added to $list_v$,
then for each $x_i \in S$,
there are at least $c_i$ entries for $x_i$ at or below $Z$ in $list_v$.
\end{lemma*}
\vspace{-.25in}
\begin{proof}
If not
 there exists an $x_i \in S$ 
with strictly less than $c_i$ entries for $x_i$ at or below $Z$ in $list_v$.

Let $Z_1^{-}$ be the $c_i$-th entry for $x_i$ in $list_y$ 
(if $x_i$ is $Z$'s source, then  $Z_1^{-}$ is $Z^-$).
If $Z_1^{-}$ is not $Z^-$, it is
 below $Z^-$ in $list_y$ and so  was sent 
 in a round $r' < r$; if $Z_1^{-} = Z^-$ then $r'=r$.
Let $Z_1$ be the corresponding entry 
created for possible addition to $list_v$ in Step~\ref{alg2:Z}  of
Algorithm~\ref{alg2}. 

If $Z_1$ was added to $list_v$ and is also present in $list_v$ in round $r$,
then by Lemma ~\ref{lemma:countDummySet} and \ref{lemma:counts},
there will be at least $c_i$ entries for $x_i$ at or below $Z_1$,
\rco.
And if $Z_1$ was removed from \lv in a round $r^{''} < r$,
then by Lemma~\ref{lemma:counts},
the number of entries for $x_i$ with key $\leq Z_1.\kappa$ should be at least $c_j$.

Now if $Z_1$ was not added to \lv in round $r'$,
then by \obs ~\ref{obs:Znotadded},
we must already have at least $c_i$ entries for $x_i$ with key $\leq Z_1.\kappa$ in round $r'$
and by Lemma~\ref{lemma:counts},
this must hold for all rounds $r^{''} > r$ as well.
\end{proof}

\begin{corollary*}{\bf\ref{cor:pos2}.}
Let $Z^-$ be an entry sent from  $y$ to $v$ in round $r$
and let $Z$ be the corresponding entry created for possible addition to $list_v$ in Step~\ref{alg2:Z}  of
Algorithm~\ref{alg2}.  
If $Z$ is added to $list_v$, 
then $pos^r_y (Z^{-}) \leq pos^r_v (Z)$.
\end{corollary*}

\subsection{Establishing an Upper bound on $Z.\nu$}

In this section (Lemmas~\ref{lemma:map}-\ref{lemma:count})
we establish an upper bound on the value of $Z.\nu$.
This upper bound on $Z.\nu$ immediately gives a bound on the
maximum number of entries that can be present in $list_v$
for a source $x \in S$.

\begin{lemma*}{\bf \ref{lemma:map}.}
Let $\mathcal{C}$ be the entries for a source $x \in S$ in $list_v$ in  round $r$. 
Then the entries in $\mathcal{C}$ can be mapped to $(d,l)$ pairs such that each
$l\geq 0$ and
each $Z \in \mathcal{C}$ is mapped to a distinct $d$ value with
$Z.\kappa = d\cdot \gamma + l$.
Also $d = d_x^*$ if $Z$ is a current shortest path entry, otherwise $d > d_x^*$.
\end{lemma*}

\begin{proof}
We will establish this result by induction on $j$, the number of entries in  $\mathcal{C}$.
For the base case, when $j=1$, we can map $d$ and $l$ to the pair in the single entry $Z$ since
$Z.\kappa = d \cdot \gamma +l$. Assume inductively that the result holds at $list_u$ for all nodes $u$ 
when the number of entries for $x$ is at most $j-1$. Consider the first time $|\mathcal{C}|$ becomes $j$
at $list_v$, and let this occur when node $y$ sends $Z^-$ to $v$ and this is updated and inserted as $Z$ 
in $list_v$ in round $r$. 

If $Z$ is inserted as a new shortest path entry with distance value $d^*$, then the distinct $d$ values currently assigned to the $j-1$ entries
for source $x$ in $list_v$ must all be larger than $d^*$ hence we can simply assign the
 $d$ and $l$ values in $Z$ as its $(d,l)$
mapping.

If $Z$ is inserted as a \dummy entry then it is possible that the $d$ value in $Z$ has already been assigned to
one of the $j-1$ entries for source $x$ on $list_v$. If this is the case, consider the entries 
for source $x$  with key value at most $Z^-.\kappa$ in $list_y$ (at node $y$). 
By the check in Step~\ref{alg2:checkIfEnough} of Algorithm~\ref{alg2} 
we know that there are $j$ such
values. Inductively these $j$ entries have $j$ distinct $d^-$ values assigned to them, and we transform these 
into $j$ distinct values for $list_v$ by adding $w'_x (y,v)\cdot \gamma + 1$ to each of them. For at least one of these
$d^-$ values in $y$, call it $d_1^-$,  it must be the case that $d'=d_1^-+ w'_x (y,v)\cdot \gamma + 1$ 
is not assigned to any of the $j-1$ entries for source $x$ below $Z$ in $list_v$. Let $Z^{-}_1$ be the entry in $y$'s list that is 
associated with distance $d_1^-$. We show that the associated $l$ value for $d'$ in $Z$ on $list_v$
must be greater than 0.

\vspace{-.2in}
\begin{align*}
(d' + w'_x (y,v))\cdot \gamma + l &= Z.\kappa \\
&= Z^-.\kappa + w'_x (y,v)\cdot \gamma + 1 \\
&\geq Z_1^-.\kappa + w'_x (y,v)\cdot \gamma + 1 \hspace{1in} \text{(since ($pos_y (Z^-) > pos_y (Z_1^-)$))} \\
&= d'\cdot \gamma + l^-_1 + w'_x (y,v)\cdot \gamma + 1 \\
&= (d' + w'_x (y,v))\cdot \gamma + l^-_1 + 1 
\end{align*}

Hence $l \geq l^-_1 + 1 > 0$.

Since $Z$ is a \dummy entry we also need to argue that $d' + w'_x (y,v) \neq d_x^*$.
If not then by induction, it implies that the entry $Z^-_{1}$ for $x$ in $y$'s list
correspond to the current shortest path entry for $x$
 in $list_y$.
Since $Z^-_{1}$ gives the shortest path distance from $x$ to $y$,
the corresponding shortest path entry for $x$ must be below $Z$ 
in $v$'s list and by induction, it must have $d_x^*$ associated with it.
This results in a contradiction since we chose the distance value, $d' + w'_x (y,v)$,
 such that it 
was different from the distances associated with the other
$(j -1)$ entries for $x$ in $v$'s list.

We have shown that the lemma holds the first time a $j$-th entry is added to $list_v$ for source $x$.
To complete the proof we now show that the
 lemma continues to hold if a new entry $Z$ for source $x$ is added to $list_v$ while keeping
the number of entries at $j$. The argument is the same as the case of having $j$ entries for source $x$ for the first time except that
we also need to consider duplication of a $d$ value at an entry above the newly inserted $Z$. For this we
proceed as in the previous case. Let $Z$ be
inserted in position $p \leq j$.  We assign a $d$ value to $Z$ as in the previous case, taking care that the $d$ value assigned to $Z$ is different from that for the
$p-1$ entries below $Z$.
Suppose $Z$'s  $d$ value has been assigned to another entry $Z''$ in $list_v$ above $Z$. Then, we consider
$Z'$, the entry that was removed (in Step 5 of {\sc Insert}) in order to keep the total number of entries for
source $x$ at $j$. We assign to
$Z''$ the value $d'$ that was assigned to $Z'$. Since $Z''$ has a larger key value than $Z'$ we will
need to use an $l''$ at least as large as that used for $Z'$ (call it $l'$)
in order satisfy the requirement that
$Z''.\kappa = d'\cdot \kappa + l''$. Since $l'$ must have been non-negative, $l''$ will also be non-negative
as required, and all $d$ values assigned to the entries for $x$ will be distinct.
\end{proof}

\begin{lemma*}{\bf \ref{lemma:belowSD}.}
Let $Z$ be the current shortest path distance entry for a source $x \in S$ in $v$'s list.
Then the number of entries for $x$ below $Z$ in $list_v$ is at most $\gamma \cdot \frac{n}{k}$.
\end{lemma*}

\vspace{-.25in}
\begin{proof}
By Lemma~\ref{lemma:map}, 
we know that the keys of all the entries for $x$ can be mapped to $(d,l)$ pairs
such that each entry is mapped to a distinct $d$ value and $l > 0$.

We have $Z.\kappa = d_x^*\cdot \gamma + l_x^*$, where $l_x^*$ is the hop-length of the shortest path from $x$ to $v$.
Let $Z^{''}$ be an entry for $x$ below $Z$ in $v$'s list. 
Then, $Z^{''}.\kappa \leq Z.\kappa$. 
It implies 

\begin{align*}
d^{''}\cdot \gamma + l^{''} &\leq d_x^*\cdot \gamma + l_x^* \\
d^{''}\cdot \gamma &\leq  d_x^*\cdot \gamma + (l_x^* -  l^{''})\\
d^{''} &\leq d_x^* + \frac{ (l_x^{*} -  l^{''})}{\gamma} \\
 &\leq d_x^* + \frac{(h -  1)}{\gamma} \\
 &< d_x^* + \frac{h}{\gamma}  \\
 &= d_x^* + \frac{h}{\gamma^2}\cdot \gamma	\\
 &= d_x^* + \frac{n}{k}\cdot \gamma 
\end{align*}

Thus $d^{''} < d_x^* + \frac{n}{k}\cdot \gamma $.
Since $d^{''} \geq d^*_x$,
there can be at most $\frac{n}{k}\cdot \gamma $  entries for $x$ below $Z$
in $list_v$.
\end{proof}

\begin{lemma*}{\bf \ref{lemma:count}.}
For each source $x \in S$, 
$v$'s list has at most $\frac{n}{k}\cdot \gamma + 1$ entries for $x$.
\end{lemma*}

\begin{proof}
On the contrary, 
let $Z$ be an entry for  source $x \in S$ with the smallest key such that
$Z$ is the $(\gamma\cdot \frac{n}{k} + 2)$-th entry for $x$ in $list_v$.
Let $y$ be the sender of $Z$ to $v$ and let the corresponding entry in $y$'s list be $Z^{-}$.

If $Z$ was added as a \dummy entry, 
then by 
Lemma~\ref{lemma:countDummySet}
there are at least $\gamma\cdot \frac{n}{k} + 2$ entries for $x$ at or below $Z^-$ in $list_y$,
resulting in a contradiction as
$Z$ is the entry with the smallest key that have this $\nu$ value.

Otherwise if $Z$ was added as a current shortest path entry,
then by Lemma~\ref{lemma:belowSD},
$Z$ can have at most $\gamma\cdot \frac{n}{k}$ entries below it
in any round
and hence 
 there are at most $\gamma\cdot \frac{n}{k} + 1$ at or below $Z$ in $list_v$
in all rounds 
(and if $Z$ is later marked as \dummy then 
by Lemma~\ref{lemma:countDummy}
$Z.\nu$ will stay fixed at that value),
again resulting in a contradiction.
\end{proof}

\subsection{Establishing an Upper Bound on the round $r$ by which an entry $Z$ is sent}
 
 \begin{lemma*}{\bf \ref	{lemma:sendBound}.}
If an entry $Z$ is added to $list_v$ in round $r$ then 
 $r  < Z.\kappa + pos^{r}_v (Z)$.
\end{lemma*}

\begin{proof}
The lemma holds in the first round since all entries have non-negative $\kappa$, any received entry has hop length at least 1, and the lowest
position is 1 so  for any entry $Z$  received by $v$ in round 1, $Z.\kappa + pos^1_v (Z)  \geq  1+1 > 1$.

Let $r$ be the first round (if any) in which the lemma is violated, and let it occur 
when entry $Z$ is added to $list_v$.
So $r \geq  Z.\kappa + pos^r_v (Z)$. 
 Let $r_1 = Z.\kappa + pos^r_v (Z)$ (so $r_1<r$ by assumption).

Since $Z$ was added to $list_v$ in round $r$,  $Z^-$ was sent to $v$ by a node $y$ 
 in round $r$. So $r=Z^-.\kappa + pos_y^r(Z^-) $. But $Z.\kappa > Z^-.\kappa$ and
 $pos_v^r(Z) \geq pos_y^r(Z^-)$, hence $r$ must be less than $Z.\kappa + pos^r_v (Z)$.
\end{proof}

\begin{lemma*}{\bf \ref{lemma:sp}.}
 Let  $\pi^*_{x,v}$ be a shortest path from source $x$ to $v$ with the minimum number of hops among
 $h$-hop
 shortest paths from $x$ to $v$. Let  $\pi^*_{x,v}$ have $l^*$ hops and
shortest path distance  $d^{*}_{x,v}$.
 Then $v$ 
receives
an entry 
  $Z^*= (\kappa, d^{*}_{x,v}, l^*, x)$ 
by round 
$r < Z^*.\kappa + pos_v^r (Z^*)$.
 \end{lemma*}
  
 \begin{proof}
 If an entry 
  $Z^*= (\kappa, d^{*}_{x,v}, l^*, x)$
  is placed on $list_v$ by $v$
 then by Lemma~\ref{lemma:sendBound}
 it is received before round  $Z^*.\kappa + pos_v^r (Z^*)$ and hence
  it will be sent in round $r = Z^*.\kappa + pos_v^r (Z^*)$
 in Step~\ref{alg2:sendStart}.
  It remains to show that an entry for path $\pi^*_{x,v}$   is received by $v$. We
 establish  this for all pairs $x,v$ by induction on key value $\kappa$. 
 
   If $\kappa = 0$, then it implies that the shortest path is the vertex $x$ itself and 
  thus the statement holds for $\kappa = 0$.
 Let us assume that the statement holds for all keys $< \kappa$ and consider the path $\pi^*_{x,v}$ with key 
 $\kappa = d^{*}_{x,v} \cdot \gamma + l^*$.
 
   Let $(y,v)$ be the last edge on the
  path $\pi^*_{x,v}$
  and let $\pi^*_{x,y}$ be the subpath of $\pi^*_{x,v}$ from $x$ to $y$.
  By construction the path $\pi^*_{x,y}$ is a shortest path from $x$ to $y$ and its hop length $l^* -1$ is the smallest
  among all shortest paths from $x$ to $y$. Hence by the inductive assumption an entry 
   $Z^{-}$ with  $Z^{-}.\kappa = d^*_{x,y} \cdot \gamma + l^* -1$ (which is strictly less than $Z^{*}.\kappa$)
   is received by $y$ before round $Z^{-}.\kappa + pos_y^{'} (Z^{-})$ (by Lemma~\ref{lemma:sendBound})
   and is then sent to $v$ in round $r' = Z^{-}.\kappa + pos_y^{r'} (Z^{-})$ in Step~\ref{alg2:sendStart}.
  Thus $v$ adds the shortest path entry for $x$, $Z^*$,
  to $list_v$ by the end of round $r'$.
   \end{proof} 
   
     \subsection{Establishing an Upper Bound on the round $r$ by which Algorithm~\ref{alg2} terminates}

    \begin{lemma*}{\bf \ref{lemma:hhopkSSP}.}
   Algorithm~\ref{alg2} correctly computes the $h$-hop shortest path distances from each source $x \in S$ to each node $v \in V$
   by round $(n-1)\gamma + h + n\cdot\gamma + k$.
   \end{lemma*}
   
   \begin{proof}
An $h$-hop shortest path has hop-length at most $h$ and 
   weight at most $n-1$, 
   hence a key corresponding to a shortest path entry 
   will have value at most $(n-1)\gamma + h$.
   Thus by Lemma~\ref{lemma:sp},
   for every source $x \in S$ every node $v \in V$
should have received the shortest path distance entry, $Z^*$,
   for source $x$ by round $r = (n-1)\gamma + h + pos_v^r (Z^*)$.
   
   Now we need to bound the value of $pos_v^r (Z^*)$.
   By Lemma~\ref{lemma:count}, 
   we know that there are at most $\gamma \cdot\frac{n}{k} + 1$ entries 
   for each source $x \in S$ in a node $v$'s list.
   Now as there are $k$ sources,
   $v$'s list has at most $(\gamma \cdot\frac{n}{k} + 1)\cdot k \leq \gamma\cdot n + k$ entries,
   thus $pos_v^r (Z^*) \leq \gamma\cdot n + k $ and
   hence $r \leq (n-1)\gamma + h + \gamma\cdot n + k$.
   \end{proof}
   
  Since $\gamma =  \sqrt{\frac{hk}{n}}$, Lemma~\ref{lemma:hhopkSSP}  establishes the
  bounds given in 
  Theorem~\ref{thm:alg2}.

\end{document}